\documentclass{article}
\usepackage{ijcai22}

\usepackage{times}
\usepackage{soul}
\usepackage{url}
\usepackage[hidelinks]{hyperref}
\usepackage[utf8]{inputenc}
\usepackage[small]{caption}
\usepackage{graphicx}
\usepackage{amsmath}
\usepackage{amsthm}
\usepackage{booktabs}
\usepackage{algorithm}
\usepackage{algorithmic}
\usepackage{amssymb}
\usepackage{multirow}

\newtheorem{theorem}{Theorem}

\title{Online Auction-Based Incentive Mechanism Design for Horizontal Federated Learning with Budget Constraint}

\author{
Jingwen Zhang 
\and
Yuezhou Wu \And 
Rong Pan
}

\begin{document}

\maketitle

\begin{abstract}
Federated learning makes it possible for all parties with data isolation to train the model collaboratively and efficiently while satisfying privacy protection. To obtain a high-quality model, an incentive mechanism is necessary to motivate more high-quality workers with data and computing power. The existing incentive mechanisms are applied in offline scenarios, where the task publisher collects all bids and selects workers before the task. However, it is practical that different workers arrive online in different orders before or during the task. Therefore, we propose a reverse auction-based online incentive mechanism for horizontal federated learning with budget constraint. Workers submit bids when they arrive online. The task publisher with a limited budget leverages the information of the arrived workers to decide on whether to select the new worker. Theoretical analysis proves that our mechanism satisfies budget feasibility, computational efficiency, individual rationality, consumer sovereignty, time truthfulness, and cost truthfulness with a sufficient budget. The experimental results show that our online mechanism is efficient and can obtain high-quality models.
\end{abstract}

\section{Introduction}
Federated Learning (FL) is a distributed machine learning framework that satisfies privacy protection, data security, and government laws~\cite{yang2019federated}. By sharing model parameters instead of data, all parties collaborate to train the model, effectively breaking down data silos~\cite{bonawitz2019towards}. Sufficient high-quality workers are the key to the success of FL. Existing researches assume that workers serve for free~\cite{yu2020fairness}. However, resource consumption and privacy leakage risks make workers reluctant to participate in FL. Restrictions on data exchange leave task publishers with no effective means of selecting and paying high-quality workers~\cite{zhang2021incentive}. Therefore, an incentive mechanism is necessary for FL.

Existing researches on incentive mechanisms~\cite{le2021incentive,zeng2020fmore} for federated learning are applied to offline scenarios, meaning that workers are selected before the task, and once the task starts, the task publisher no longer accepts bids from newly arrived workers. However, it is practical that workers arrive online at different times~\cite{zhao2014crowdsource}. Waiting for all interested workers to arrive before the task can result in task delay and impairment of benefit.

We design an online reverse auction-based incentive mechanism for horizontal federated learning with budget constraint and leverage the reputation proposed by Zhang~\shortcite{zhang2022auctionbased} to indirectly reflect the quality and reliability of workers. We divide FL task into multiple time steps according to global iterations. Workers arrive at different time steps in different orders before or during tasks. At the first time step, the publisher selects arrived workers by combining their bids and reputations through the offline proportional share mechanism~\cite{singer2010budget}. At other time steps, arrived workers are divided into two groups, which are used to mutually estimate payment density thresholds. Workers whose unit reputation bid prices do not exceed the corresponding payment density threshold are selected. Theoretical analysis proves that our mechanism satisfies budget feasibility, computational efficiency, individual rationality, consumer sovereignty, time truthfulness and cost truthfulness with sufficient budget. The experimental results show that it helps to obtain high-quality models.

The rest of the paper is organized as follows. Section~\ref{sec:related} introduces related work. Section~\ref{sec:system_model_pro_def} describes the system model and problem definition. Section~\ref{sec:online_mechanism} designs the mechanism in detail. Section~\ref{sec:analysis} conducts theoretical analysis and Section~\ref{sec:experiments} shows the simulation results.

\section{Related Work}
\label{sec:related}
Jiao \textit{et al.}~\shortcite{jiao2020toward} develops a mechanism based on multi-dimensional reverse auction, which takes into account the data quantity and distribution of workers. Ying \textit{et al.}~\shortcite{ying2020double} proposes a framework SHIELD based on reverse auction and differential privacy. The probability of being selected is based on the bid price. Roy \textit{et al.}~\shortcite{roy2021distributed} considers the cost, QoE, and reputation to select workers through reverse auction and reputation. Zhang \textit{et al.}~\shortcite{zhang2021incentive} prioritizes workers with lower unit reputation bid prices. Zhou \textit{et al.}~\shortcite{zhou2021truthful} selects workers using a greedy algorithm with the goal of minimizing social costs through the reverse auction. Seo \textit{et al.}~\shortcite{seo2021auction} proposes an auction-based method to help task publishers dynamically select resource-efficient workers. To improve model accuracy and communication efficiency, Pandey \textit{et al.}~\shortcite{pandey2020crowdsourcing} models the interaction between the publisher and workers as a two-stage Stackelberg game. Feng \textit{et al.}~\shortcite{feng2019joint} uses Stackelberg game. Workers decide the unit price of data, while the publisher decide the amount of data. Ding \textit{et al.}~\shortcite{ding2020incentive} uses a multi-dimensional contract as an incentive mechanism, taking into account training costs and communication delays. Lim \textit{et al.}~\shortcite{lim2021towards} establishes multiple contract items, and the worker selects a contract item and gets paid. All of the above are applied to offline scenarios rather than online scenarios.

\section{System Model and Problem Definition}
\label{sec:system_model_pro_def}
\subsection{System Model}
A federated learning system includes a task publisher and many workers. Since his data is not enough, the publisher recruits workers to train a high-quality model through FL. The data of the task publisher can be used as the validation set and test set. Workers may be users of smart devices with data and computing power. Due to daily life, workers have accumulated lots of data. The publisher with a budget of $B$ publishes a FL task with $T$ rounds global iterations and scheduled start time. Each interested worker $i$ that arrives at different times, submits a sealed bid price to the task publisher. Each worker $i$ has a true arrival time $a_i \in \{1,...,T\}$, the true cost $c_i$ of a global iteration, as well as the quantity and quality of data, but all are private. Only the budget $B$, the number of global iterations $T$, and the reputation $Re_i$ are public.

We model the interaction between the publisher and workers as an online reverse auction model. At his arrival time $\hat{a_i}$, worker $i$ submits the bid price $b_i$ to the publisher, which represents the claimed cost of a global iteration. Since in reality the worker is strategic and tries to maximize his utility, his arrival time $\hat{a_i}$ and bid price $b_i$ may be different from his true arrival time $a_i$ and cost $c_i$. We assume that workers will not claim an arrival time $\hat{a_i}$ earlier than the true arrival time $a_i$, that is, $\hat{a_i} \geq a_i$. That is because if a worker lied about an earlier arrival time and was selected, the publisher could find that he did not upload a local model. We also assume that the workers are independent and cannot collude with each other. Due to the uneven abilities, the quantity and quality of data of different workers may be different. To obtain a higher-quality model with a limited budget, the publisher will select high-quality workers from the workers who arrive online. Since the reputation $Re_i$ indirectly reflects the quality and reliability of the worker, the task publisher selects the workers and decides to pay online based on the bid price and reputation.

\begin{figure}[htbp]
  \centering
  \centerline{\includegraphics[width=\linewidth]{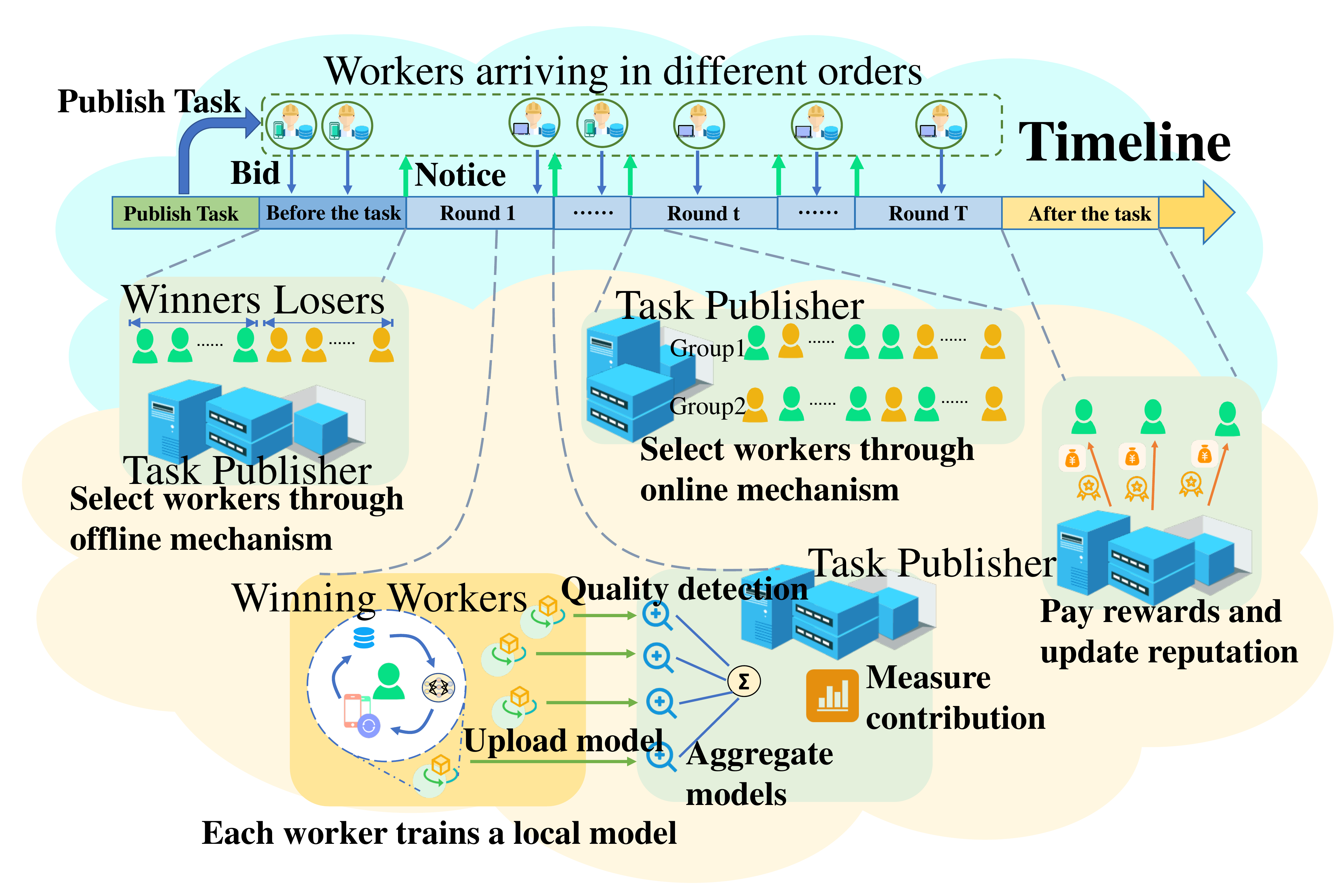}}
  \caption{System Model}
  \label{fig:system_model}
\end{figure}
Figure~\ref{fig:system_model} shows the workflow. First, the publisher publishes a FL task, announcing the budget $B$, the number of global iterations $T$, the scheduled start time, the model structure, data type, and other requirements. According to the strategy, the interested workers submit their arrival time $\hat{a_i}$ and bid price $b_i$ to the publisher one by one online. At the scheduled start time, the task publisher collects all bids of workers with $\hat{a_i} = 1$ and selects them through the offline reverse auction mechanism by combining the bid prices and reputation. If the number of selected workers is less than the minimum number required to start the task, the publisher will delay the start time and re-execute the worker selection phase at the new start time. The winning workers participate in all remaining global iterations. The unselected workers will not leave but wait for the next worker selection phase before the next global iteration. The winning worker starts a round of global iteration. The task publisher notifies the winning workers and distributes the initial global model to them. The winning worker uses his data to train the local model based on the downloaded global model and then uploads the local model. The publisher aggregates all local models to obtain a new global model. At this point, a global iteration is over. Before each remaining global iteration $t> 1$, the publisher collects the bids of newly arrived workers and previously unselected workers, combined with the reputation, and selects new ones through an online reverse auction mechanism with budget constraint. Then the global iteration $t$ starts. After $T$ global iterations, the task publisher updates the reputation of the winning worker $i$ and pays him $p_i$.

\subsection{Problem Definition}
To select cost-effective workers online with a limited budget $B$ and obtain a higher-quality model after $T$ global iterations, it is necessary to design an incentive mechanism $\mathbb{M}(\mathbf{f},\mathbf{p})$, including an online selection mechanism $\mathbf{f}$ and a payment mechanism $\mathbf{p}$. Suppose the selected workers form the set $S$. Define the selected time step of the worker $i$ as $t_i = t$, if the worker $i$ is selected at the $t$-th time step. The utility $u_i$ of the worker $i$ is
\begin{equation}
    u_i = 
    \left\{
        \begin{aligned}
        0  &, & i \notin S, \\
        p_i - c_i \cdot (T - t_i + 1) & , & i \in S.
        \end{aligned}
    \right.
\end{equation}
Since the budget of the task publisher is limited, $\sum_{i \in S}p_i \leq B$ needs to be met. The task publisher hopes to select as many high-quality workers as possible with the limited budget. Since reputation indirectly reflects the quality and reliability of workers, the utility $\mathcal{U}$ of the publisher is
\begin{equation}
    \mathcal{U} = \sum_{i \in S} Re_i(T - t_i + 1).
\end{equation}

The online mechanism $\mathbb{M}$ is designed to maximize the utility $\mathcal{U}$ of the publisher by determining the set $S$ with the budget constraint. At the same time, the mechanism needs to meet the following six critical economic properties.
\begin{itemize}
    \item \textbf{Individual Rationality}: The worker cannot be paid less than his true cost.
    \item \textbf{Budget Feasibility}: The sum of the rewards paid by the task publisher to the workers cannot exceed the budget.
    \item \textbf{Computational Efficiency}: The time complexity of the mechanism is polynomial.
    \item \textbf{Consumer Sovereignty}: The publisher cannot arbitrarily exclude any worker. As long as their bids are low enough, workers have chances to be selected and paid.
    \item \textbf{Cost Truthfulness with Sufficient Budget}: When the remaining budget is sufficient, the worker can maximize his utility by submitting his true cost.
    \item \textbf{Time Truthfulness with Sufficient Budget}: When the remaining budget is sufficient, the worker can maximize his utility by submitting his true arrival time.
\end{itemize}

\section{Online Incentive Mechanism Design}
\label{sec:online_mechanism}
Designing an online incentive mechanism for federated learning requires overcoming many challenges. The worker may lie to the publisher about his arrival time and cost for greater benefit. Thus, the mechanism should be designed to motivate workers to report true information, which helps the publisher make decisions. Moreover, the total rewards paid to workers cannot exceed the budget. Furthermore, the mechanism needs to be applied to the scenario where workers arrive online in different orders. In the offline scenario, the publisher has collected bids submitted by all workers, so it is easy to make a selection decision. In the online scenario, the publisher cannot collect bids from all workers in advance but only receives online bids from different workers before or during the task. With incomplete information, it is difficult to make a decision on whether to select the arrived worker or not.

In existing incentive mechanisms for federated learning or others~\cite{deng2021fair}, a certain indicator of selected workers will not exceed a payment density threshold. For example, in RRAFL~\cite{zhang2021incentive} the payment density threshold is $\frac{b_{k+1}}{Re_{k+1}}$, while in the proportional share mechanism~\cite{singer2010budget} it is $\min(\frac{B}{\sum_{i \in S}Re_i}, \frac{b_{k+1}}{Re_{k+1 }})$. Inspired by this, in the online scenario, the publisher can make decisions by learning a payment density threshold.

A two-stage sampling-accepting process is utilized to learn the payment density threshold~\cite{babaioff2008online}. The first stage rejects workers and only collects their bids as samples to learn the threshold. The second stage leverages this threshold for worker selection. However, this mechanism does not satisfy consumer sovereignty since the workers who bid in the first stage cannot win no matter how they bid. Workers are more inclined to delay arrivals, which will lead to task starvation. In addition, since the threshold is only learned using samples from the first stage and not updated in later stages, it may lack accuracy.

Zhao \textit{et al.}~\shortcite{zhao2014crowdsource} proposes a multi-stage sampling-accepting process to solve the above problem. The task is divided into multiple stages. At each stage, bids from workers who have left in previous stages are added to the sample set. The payment density threshold is updated by dynamically increasing the sample set size $U^{\prime}$ and the sample budget $B^{\prime}$. This approach satisfies consumer sovereignty without causing task starvation while making the threshold more accurate by dynamically updating. However, we assume that once the worker submits a bid, whether he is selected or not, he will stay until the end of the task. There is no way to update the threshold by adding departing workers to the sample set as above. Therefore, this approach cannot be applied directly.

In federated learning, the workers selected at different stages participate in different numbers of global iterations. Only a sufficient number of workers are selected, can the FL task start. In the crowdsensing task, workers selected in different stages perform the same independent task. Their workload is the same. Moreover, workers can work immediately when they are selected. The incentive mechanism for crowdsensing cannot be directly applied to federated learning. 

We propose an online incentive mechanism $\mathbb{M}(\mathbf{f}, \mathbf{p})$ according to the nature of FL and design a method to learn the payment density threshold. To obtain a higher-quality model with a limited budget $B$, on the one hand, the publisher selects more workers with lower bid prices $b_i$, and on the other hand, selects higher-quality workers with higher reputation $Re_i$. Balancing the worker's bid price $b_i$ and the reputation $Re_i$, define the cost density $\rho_i$ of worker $i$ as
\begin{equation}
    \rho_i = \frac{b_i}{Re_i}.
\end{equation}
\begin{algorithm}[b]
    \caption{GetPaymentDensityThreshold}
    \label{alg:get_density_threshold}
    \textbf{Input}: Sample budget $B^{\prime}$; Sample worker set $U^{\prime}$;\\
    \textbf{Output}: Density threshold $\rho^{\ast}$;
    \begin{algorithmic}[1]
        \STATE Sort the workers in $U^{\prime}$ such that $\frac{b_1}{Re_1} \leq ... \leq \frac{b_{|U^{\prime}|}}{Re_{ |U^{\prime}|}}$;
        \label{alg:get_density_threshold:sort}
        \STATE $S^{\prime} = \phi$; i = 1; 
        \WHILE {$\frac{b_i}{Re_i} \leq \frac{B^{\prime}}{Re_i + \sum_{j \in S^{\prime}}Re_j}$}
        \label{alg:get_density_threshold:choose_start}
            \STATE $S^{\prime} = S^{\prime} \cup \{i\}$; $i = i + 1$;
        \ENDWHILE
        \label{alg:get_density_threshold:choose_end}
        \STATE $k = i - 1$;
        \STATE $\rho^{\ast} = \min(\frac{B^{\prime}}{\sum_{j \in S^{\prime}} Re_j}, \frac{b_{k+1}}{Re_{k+1}} )$;
        \STATE \textbf{Return} $\rho^{\ast}$;
    \end{algorithmic}
\end{algorithm}
We naturally use the offline proportional share mechanism~\cite{singer2010budget} to find out the payment density threshold based on the sample set and sample budget, as shown in Algorithm \ref{alg:get_density_threshold}. Algorithm \ref{alg:get_density_threshold} adopts a greedy strategy, which first arranges the workers in the sample set $U^{\prime}$ in the order of increasing cost density, that is $\rho_1 \leq ... \leq \rho_k \leq \rho_{k+1} \leq ... \leq \rho_{|U^{\prime}|}$. According to the proportional share allocation rule, we find the last worker $k$ in the sequence that satisfies $\rho_k \leq \frac{B^{\prime}}{Re_k + \sum_{i = 1}^{k-1}Re_i}$ from front to back. The first $k$ workers in the sequence form the set $S^{\prime}$. The learned payment density threshold is $\rho^{\ast} = \min(\frac{B^{\prime}}{\sum_{i \in S^{\prime} }Re_i}, \frac{b_{k+1}}{Re_{k+1}})$.
\begin{algorithm}[tb]
    \caption{SelectWorkersFromGroup}
    \label{alg:online_auction_for_each_group}
    \textbf{Input}: Worker set $U$; Payment density threshold $\rho^{\ast}$; Winner Worker set $S$; 
    \begin{algorithmic}[1]
        \STATE Sort all workers in $U$ such that $Re_1 \leq ... \leq Re_{|U|}$;
        \label{alg:online_auction_for_each_group:sort_set}
        \STATE $i = 1$;
        \WHILE {$ i \leq |U|$}
        \label{alg:online_auction_for_each_group:choose_from_set:start}
            \STATE $\rho_i = \frac{b_i}{Re_i}$;
            \IF {$\rho_i \leq \rho^{\ast}$}
                \IF {$i \notin S$ and $(T-t+1)\cdot Re_i\cdot \rho^{\ast} \leq \frac{B}{2} - \sum_{j \in U}p_j$}
                \label{alg:online_auction_for_each_group:choose_unchosen_worker:start}
                    \STATE $S = S \cup \{i\}$;
                    $p_i = (T-t+1)Re_i \rho^{\ast}$;
                    $\rho^{\ast}_{i} = \rho^{\ast}$;
                \label{alg:online_auction_for_each_group:choose_unchosen_worker:end}
                \ELSIF {$i \in S$}
                \label{alg:online_auction_for_each_group:choose_chosen_worker:start}
                    \STATE $p_i^{\prime} = p_i + (\rho^{\ast} - \rho_{i}^{\ast}) \cdot Re_i \cdot (T - t + 1)$;
                    \label{alg:online_auction_for_each_group:update_payment:start}
                    \IF {$p_i^{\prime} > p_i$}
                        \IF {$p_i^{\prime} > \frac{B}{2} - \sum_{j \in U}p_j + p_i$}
                            \STATE $p_i^{\prime} = \frac{B}{2} - \sum_{j \in U}p_j + p_i$;
                        \ENDIF
                        \STATE $p_i = p_i^{\prime}$;
                        $\rho_i^{\ast} = \rho^{\ast}$;
                    \ENDIF
                    \label{alg:online_auction_for_each_group:update_payment:end}
                \ENDIF
                \label{alg:online_auction_for_each_group:choose_chosen_worker:end}
            \ENDIF 
            \STATE $i = i + 1$;
        \ENDWHILE
        \label{alg:online_auction_for_each_group:choose_from_set:end}
    \end{algorithmic}
\end{algorithm}
\begin{algorithm}[tb]
    \caption{Online Selection And Payment Mechanism}
    \label{alg:online_auction}
    \textbf{Input}: Budget $B$; Total rounds of global iterations $T$; First-round budget ratio $ratio$;
    \begin{algorithmic}[1]
        \STATE \textbf{/*Before the task.*/}
        \STATE $B_1 = B \cdot ratio$; $t = 1$; $U_1 = \phi$; $U_2 = \phi$;
        \WHILE {The scheduled start time has not been reached}
        \label{alg:online_auction:join_in:start1}
            \STATE Worker $i$ arrives; $p_i = 0$; 
            \STATE Add worker $i$ to $U_1$ or $U_2$ according to his $gid_i$.
        \ENDWHILE
        \label{alg:online_auction:join_in:end1}
        \STATE \textbf{/*Select workers at the first time step through an offline reverse auction.*/}
        \STATE $U = U_1 \cup U_2$; $S = \phi$; $i = 1$;
        \STATE Sort the workers in $U$ such that $\frac{b_1}{Re_1} \leq ... \leq \frac{b_{|U|}}{Re_{|U|}}$;
        \label{alg:online_auction:first_sort}
        \WHILE {$\frac{T\cdot b_i}{Re_i} \leq \frac{B_1}{Re_i + \sum_{j \in S}Re_j}$}
        \label{alg:online_auction:first_choose:start}
            \STATE $S = S \cup \{i\}$; $i = i + 1$;
        \ENDWHILE
        \label{alg:online_auction:first_choose:end}
        \STATE $k = i - 1$; $\rho^{\ast} = \frac{\min(\frac{B_1}{\sum_{j \in S}Re_j}, \frac{T \cdot  b_{k+1}}{Re_{k+1}})}{T}$;
        \label{alg:online_auction:first_rho}
        \FOR {each worker $i \in S$}
        \label{alg:online_auction:first_payment:start}
            \STATE $p_i = T \cdot Re_i \cdot \rho^{\ast}$; $\rho^{\ast}_i = \rho^{\ast}$;
        \ENDFOR
        \label{alg:online_auction:first_payment:end}
        \STATE \textbf{/*Select workers at other time steps through an online reverse auction.*/}
        \STATE $t = 2$;
        \WHILE {$t \leq T$}
            \STATE Add each newly worker $i$ arrived in $(t-1, t]$ to $U_1$ or $U_2$ according to his $gid_i$, and set $p_i = 0$;
            \label{alg:online_auction:join_in:end2}
            \STATE $B^{\prime} = (B_1 + \frac{(B-B_1)(t-1)}{T-1})/T$;
            \STATE $\rho^{\ast}_{U_1} = GetPaymentDensityThreshold(\frac{B^{\prime}}{2}, U_1)$;
            \label{alg:online_auction:get_threshold1}
            \STATE $\rho^{\ast}_{U_2} = GetPaymentDensityThreshold(\frac{B^{\prime}}{2}, U_2)$;
            \label{alg:online_auction:get_threshold2}
            \STATE $SelectWorkersFromGroup(U_1, \rho_{U_2}^{\ast}, S)$;
            \label{alg:online_auction:choose_from_set1}
            \STATE $SelectWorkersFromGroup(U_2, \rho_{U_1}^{\ast}, S)$;
            \label{alg:online_auction:choose_from_set2}
            \STATE $t = t + 1$;
        \ENDWHILE
    \end{algorithmic}
\end{algorithm}

We design the online selection mechanism $\mathbf{f}$ and payment mechanism $\mathbf{p}$ using a modified multi-stage sampling-accepting process, as shown in Algorithm \ref{alg:online_auction}. Federated learning consists of multiple rounds ($T$ rounds) of global iterations, which correspond to multiple stages (multiple time steps). Since only a sufficient number of workers are selected can the FL task start, we design different worker selection methods for the first time step and other time steps, respectively. The first-round budget ratio is $ratio$, and the first-round budget is $B_1 = B \cdot ratio$. The task publisher sets up two arrived worker groups $U_1$ and $U_2$, which is used as the sample set for each other to calculate the payment density threshold. A worker's bid does not affect the payment density threshold used to decide on him, thus ensuring cost truthfulness. The result of $id_i \% 2$ is as $gid_i$ of worker $i$, which determines which group to join. Workers cannot join other groups by modifying their bids. Because of the randomness of the hash value, the distribution of bids in the two groups is similar, thus ensuring the validity of the payment density threshold. 

At the first time step $t = 1$, once the scheduled start time is reached, the task publisher selects workers through the offline proportional share mechanism. First arranges the workers in the arrived workers set $U$ in the order of increasing cost density,that is
\begin{equation}
    \rho_1 \leq ... \leq \rho_k \leq \rho_{k+1} \leq ... \rho_{|U|}.
\end{equation}
Then, find the last worker $k$ in the worker sequence that satisfies $T\rho_k \leq \frac{B_1}{Re_k + \sum_{i = 1}^{k}Re_i}$. If the number of workers selected is lower than the minimum number to start the first global iteration, delay the start time and repeat the above process. Otherwise, the first $k$ workers in the sequence form the winning worker set $S$ to participate in all global iterations. According to the worker selection result of the first time step $t = 1$, the payment density threshold $\rho^{\ast}$ is 
\begin{equation}
    \rho^{\ast} = \min(\frac{B_1}{T \sum_{j \in S}Re_j}, \frac{b_{k+1}}{Re_{k+1}}).
\end{equation}
Unselected workers are paid 0 and wait for the next worker selection process at the next time step. The reward for the selected worker $i \in S$ is temporarily $p_i = T \cdot Re_i \cdot \rho^{\ast}$. The maximum payment density threshold he encounters after being selected is $\rho_i^{\ast} = \rho^{\ast}$. At other time steps $t > 1$, all newly arrived workers join the corresponding group, and their rewards are initialized to 0. First, update the sample budget $B^{\prime} = (B_1 + \frac{(B- B_1)(t - 1)}{T-1})/T$. Then, with $U_1$ and $U_2$ as the sample set respectively, and $\frac{B^{\prime}}{2}$ as the sample budget, compute the payment density threshold $\rho_{U_1}^{\ast}$ and $\rho_{U_2}^{\ast}$ through Algorithm~\ref{alg:get_density_threshold}. $\rho_{U_2}^{\ast}$ and $\rho_{U_1}^{\ast}$ are used to make decisions on workers in $U_1$ and $U_2$, respectively. Both $U_1$ and $U_2$ have budgets of $\frac{B}{2}$. For $U_1$, the task publisher prioritizes decisions on workers with higher reputation. There are both unselected and selected workers in $U_1$. If the unselected worker $i$ satisfies $\rho_i \leq \rho_{U_2}^{\ast}$ and $U_1$ has enough remaining budget, then he is selected and added to $S$. He will participate in the remaining $(T - t + 1)$ rounds, $p_i = (T - t + 1) \cdot Re_i \cdot \rho_{U_2}^{\ast} $ and $\rho_i^{\ast} = \rho_{U_2}^{\ast}$. If the selected worker $i$ satisties $\rho_i \leq \rho_{U_2}^{\ast}$ and $\rho_i^{\ast} < \rho_{U_2}^{\ast}$, and $U_1$ has remaining budget, then $p_i = \min(p_i + (\rho_{U_2}^{\ast} - \rho_i^{\ast})\cdot Re_i \cdot (T - t + 1), \frac{B}{2} - \sum_{j \in U_1}p_j + p_i)$ and $\rho_i^{\ast} = \rho_{U_2}^{\ast}$. Do the same for $U_2$ as $U_1$, using $\rho_{U_1}^{\ast}$ instead of $\rho_{U_2}^{\ast}$. Workers in $S$ start the next global iteration.

\section{Theoretical Analysis}
\label{sec:analysis}
We will prove that our mechanism satisfies individual rationality, budget feasibility, computational efficiency, consumer sovereignty, cost truthfulness and time truthfulness with a sufficient budget.
\begin{theorem}
    The mechanism satisfies individual rationality.
\end{theorem}
\begin{proof}
If worker $i$ is selected at $t=1$, indicating $b_i \leq \min(\frac{B}{\sum_{j \in S} Re_j}, \frac{T \cdot b_{k+1 }}{Re_{k+1}}) \cdot Re_i / T$, his temporary reward is $p_i = \min(\frac{B}{\sum_{j \in S} Re_j}, \frac{T \cdot b_{k+1}}{Re_{k+1}}) \cdot Re_i \geq T\cdot b_i$. If worker $i$ is selected at $t > 1$, indicating $\frac{b_i}{Re_i} \leq \rho^{\ast}$, his temporary reward is $p_i = \rho^{\ast} \cdot Re_i \cdot (T - t + 1) \geq \frac{b_i}{Re_i} \cdot Re_i \cdot (T - t + 1) = b_i(T - t + 1)$. The final reward of the selected worker is not lower than the temporary reward, indicating $u_i \geq 0$. The utility of unselected workers is $u_i = 0$.
\end{proof}
    
\begin{theorem}
    The mechanism satisfies budget feasibility.
\end{theorem}
\begin{proof}
At $t=1$, the total temporary reward is $\sum_{i \in S}p_i = \sum_{i \in S}Re_i \min(\frac{B_1}{\sum_{j \ in S}Re_j}, \frac{T\cdot b_{k+1}}{Re_{k+1}}) \leq \sum_{i \in S}\frac{B_1 Re_i}{\sum_{ j \in S}Re_j} = B_1$. At $t > 1$, lines \ref{alg:online_auction:choose_from_set1} and \ref{alg:online_auction:choose_from_set2} of Algorithm~\ref{alg:online_auction} ensure that the total reward of workers selected from $U_1$ and $U_2$ does not exceed $\frac{B}{2}$ respectively, which means the total reward does not exceed $B$.
\end{proof}

\begin{theorem}
    The mechanism satisfies computational efficiency.
\end{theorem}
\begin{proof}
For Algorithm~\ref{alg:get_density_threshold}, the time complexity of sorting the workers in $U^{\prime}$ (line \ref{alg:get_density_threshold:sort}) is $O(|U^{\prime}|\log_2 |U^{\prime}|)$ and that of selecting workers from $U^{\prime}$ (lines \ref{alg:get_density_threshold:choose_start}-\ref{alg:get_density_threshold:choose_end}) is $O(|U^{\prime}|)$. Therefor, the time complexity of Algorithm \ref{alg:get_density_threshold} is $O( |U^{\prime}|\log_2 |U^{\prime}|)$. For Algorithm~\ref{alg:online_auction_for_each_group}, the time complexity of sorting the workers (line \ref{alg:online_auction_for_each_group:sort_set}) is $O(|U|\log_2 |U|)$ and that of selecting or updating the payment (lines \ref{alg:online_auction_for_each_group:choose_from_set:start}-\ref{alg:online_auction_for_each_group:choose_from_set:end}) is $O(|U|)$. Therefore the time complexity of Algorithm~\ref{alg:online_auction_for_each_group} is $O(|U|\log_2 |U|)$. Suppose the number of workers bidding during $(t-1, t]$ is $n_t$. At $t=1$, the time complexity of sorting the workers (line \ref{alg:online_auction:first_sort}) is $O(n_1\log_2 n_1)$, that of selecting workers (lines \ref{alg:online_auction:first_choose:start}-\ref{alg:online_auction:first_choose:end}) is $O(| S_1|)$ and that of calculating the temporary reward (lines \ref{alg:online_auction:first_payment:start}-\ref{alg:online_auction:first_payment:end}) is $O(n_1)$. At $t > 1$, the time complexity of computing the payment density thresholds for $U_1$ and $U_2$ (lines \ref{alg:online_auction:get_threshold1} and \ref{alg:online_auction:get_threshold2}) is $O( (\sum_{j=1}^t n_j)\log_2 (\sum_{j=1}^t n_j))$ and that of making a decision on $U_1$ and $U_2$ is $O((\sum_{j=1}^t n_j)\log_2 (\sum_{j=1}^t n_j))$. The time complexity of Algorithm~\ref{alg:online_auction} is $O(\sum_{t=1}^T ((\sum_{j=1}^t n_j)\log_2 (\sum_{j=1}^t n_j)))$.
\end{proof}

\begin{theorem}
    The mechanism satisfies consumer sovereignty.
\end{theorem}
\begin{proof}
At each time step, as long as the worker's bid is low enough and the remaining budget is sufficient, that worker is selected and paid. The task publisher does not automatically reject either worker. Therefore, the mechanism satisfies consumer sovereignty.
\end{proof}

\begin{theorem}
    The mechanism satisfies cost truthfulness with a sufficient budget.
\end{theorem}
\begin{proof}
Consider the scenario where the worker wins with $c_i$. \\
\textbf{Case 1}: If worker $i$ with $c_i$ is selected at $t \geq \hat{a_i}$ and  he with $b_i$ is still selected at $t$, then $u(b_i, b_{-i}) = u(c_i, b_{-i})$. \\
\textbf{Case 2}: Suppose worker $i$ with $c_i$ is selected at $t > \hat{a_i}$. If he with $b_i < c_i$ is selected at $t^{\prime} < t$, then he will get extra reward $Re_i \cdot \rho^{\ast}$ at each time step during $[t^{\prime}, t)$ ($\rho^{\ast}$ may change at different time steps). Worker $i$ with $c_i$ is not selected in $[t^{\prime}, t)$, indicating that $\frac{c_i}{Re_i} > \rho^{\ast}$, that is, $c_i > Re_i \cdot \rho^{\ast}$. This means that the utility of worker $i$ with $b_i$ is negative during $[t^{\prime}, t)$. Thus $u(b_i, b_{-i}) < u(c_i, b_{-i })$. \\
\textbf{Case 3}: If worker $i$ with whether $b_i$ does not win, then $u(c_i, b_{-i}) \geq u(b_i, b_{-i}) = 0$. \\
\textbf{Case 4}: Suppose worker $i$ with $c_i$ is selected at $t$ and he with $b_i > c_i$ is selected at $t^{\prime} > t$. Worker $i$ with $c_i$ will get extra reward $Re_i \cdot \rho^{\ast}$ at each time step during $[t, t^{\prime})$ ($\rho^{\ast}$ may change at different time steps). Since $\frac{c_i}{Re_i} \leq \rho^{\ast}$, that is, $c_i \leq Re_i \cdot \rho^{\ast}$. This means that the utility of worker $i$ with $c_i$ is not negative during $[t, t^{\prime})$. Thus $u(b_i, b_{-i}) \leq u(c_i, b_{-i})$. 

Then, consider the scenario where the worker with $c_i$ does not win at each time step. \\
\textbf{Case 1}: If worker $i$ with $b_i$ does not win, then $u_i(b_i, b_{-i}) = u_i(c_i, b_{-i}) = 0$. \\
\textbf{Case 2}: If worker $i$ with $b_i < c_i$ is selected at $t$, then he will get reward $Re_i \cdot \rho^{\ast}$ at each time step during $[t, T]$ ($\rho^{\ast}$ may change at different time steps). Worker $i$ with $c_i$ is not selected, indicating that $\frac{c_i}{Re_i} > \rho^{\ast}$, that is, $c_i > Re_i \cdot \rho^{\ast}$. This means that $u(b_i, b_{-i}) < u(c_i, b_{-i }) = 0$. 

In summary, the mechanism satisfies cost truthfulness.
\end{proof}

\begin{theorem}
    The mechanism satisfies time truthfulness with a sufficient budget.
\end{theorem}
\begin{proof}
If a worker is selected at time step $t$, he will be paid a maximum reward that can be earned during $[t, T]$. If reporting the true arrival time $a_i$ makes worker $i$ to be selected at time step $t$, and reporting a later arrival time $\hat{a_i}$ makes him to be selected at time step $t^{\prime}$, then $t \leq t^{\prime}$. Reporting $a_i$ allows worker $i$ to receive an extra reward during $[t, t^{\prime}]$, but reporting $\hat{a_i}$ does not. As long as $(a_i, c_i)$ are submitted, the utility of worker $i$ is greater than that of reporting $\hat{a_i}$. Therefore, the mechanism satisfies time truthfulness with a sufficient budget.
\end{proof}

\section{Experiments}
\label{sec:experiments}
We use the MNIST dataset with a two-layer fully connected model which has a hidden layer of 50 cells, and the Fashion MNIST dataset with a LeNet model for experiments. The task publisher has a validation set and a test set of size 5000 each. Each worker has a training set of size 1000. The dataset is iid, but its accuracy may not be the same, which is achieved by modifying the labels to others. Workers train the model for 1 epoch with a learning rate of 0.05 and a batch size of 128 at each global iteration . Each task contains 10 global iterations. We set the minimum number of workers to start a task to be 1, and the first-round budget ratio to be 0.35.

We set up multiple benchmarks. The first, Fixed Threshold, is online using a fixed threshold $\rho^{\ast} = 0.75$. The second is RRAFL, proposed by Zhang \textit{et al.}~\shortcite{zhang2021incentive}. The third, Vanilla FL, randomly selects workers as long as the budget remains. The fourth, Bid Greedy, prefers workers with low bids. The fifth is Proportional Share, proposed by Singer \textit{et al.}~\shortcite{singer2010budget}. The sixth is Approximate Optimal, where the task publisher has prior information on all workers, such as $c_i$. The above five are offline mechanisms.

\begin{figure}[bt]
  \centering
  \centerline{\includegraphics[width=\linewidth]{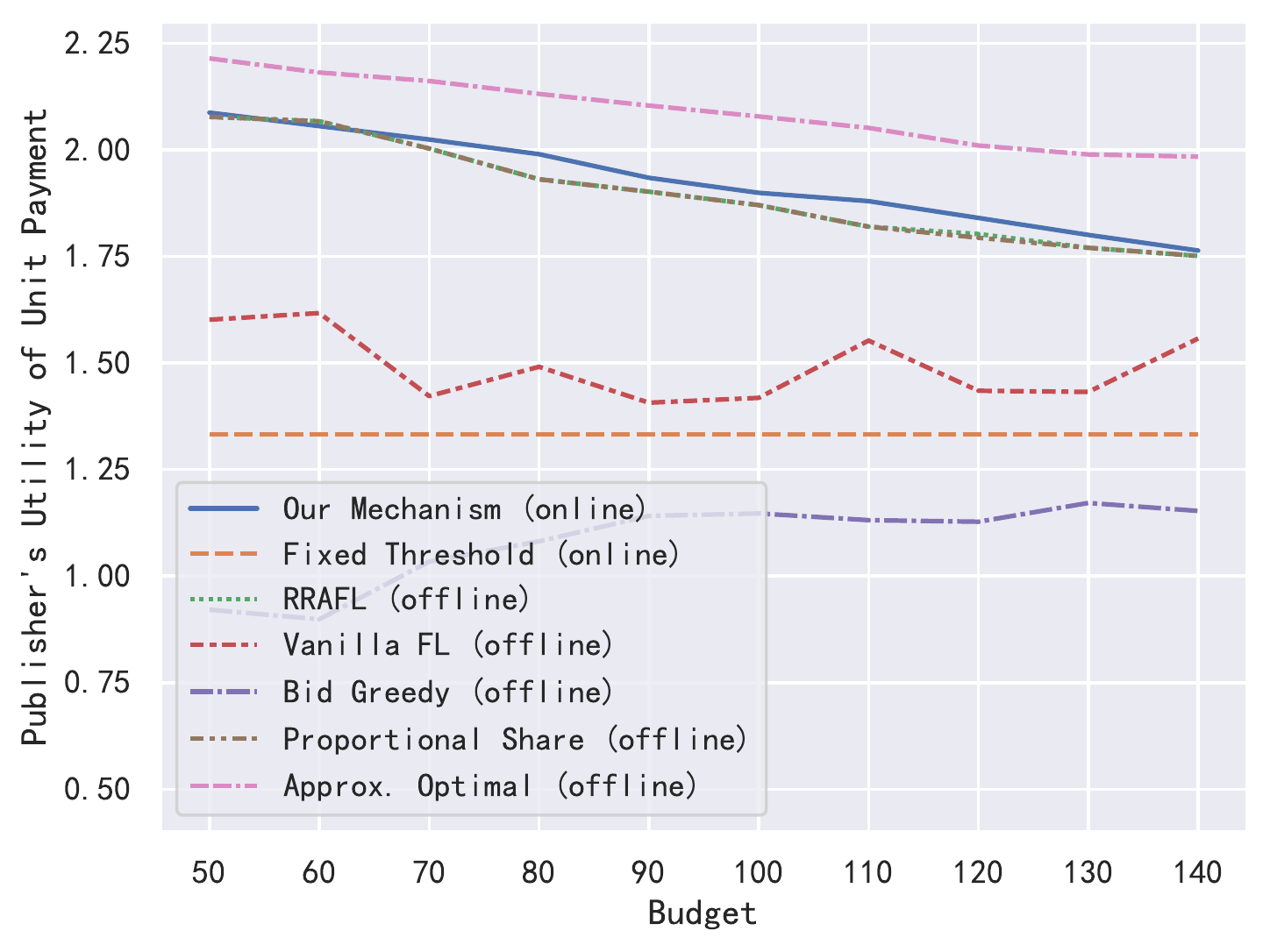}}
  \caption{Impact of budget on the publisher's utility of unit payments}
  \label{fig:budget_compare}
\end{figure}

\begin{figure}[bt]
  \centering
  \centerline{\includegraphics[width=\linewidth]{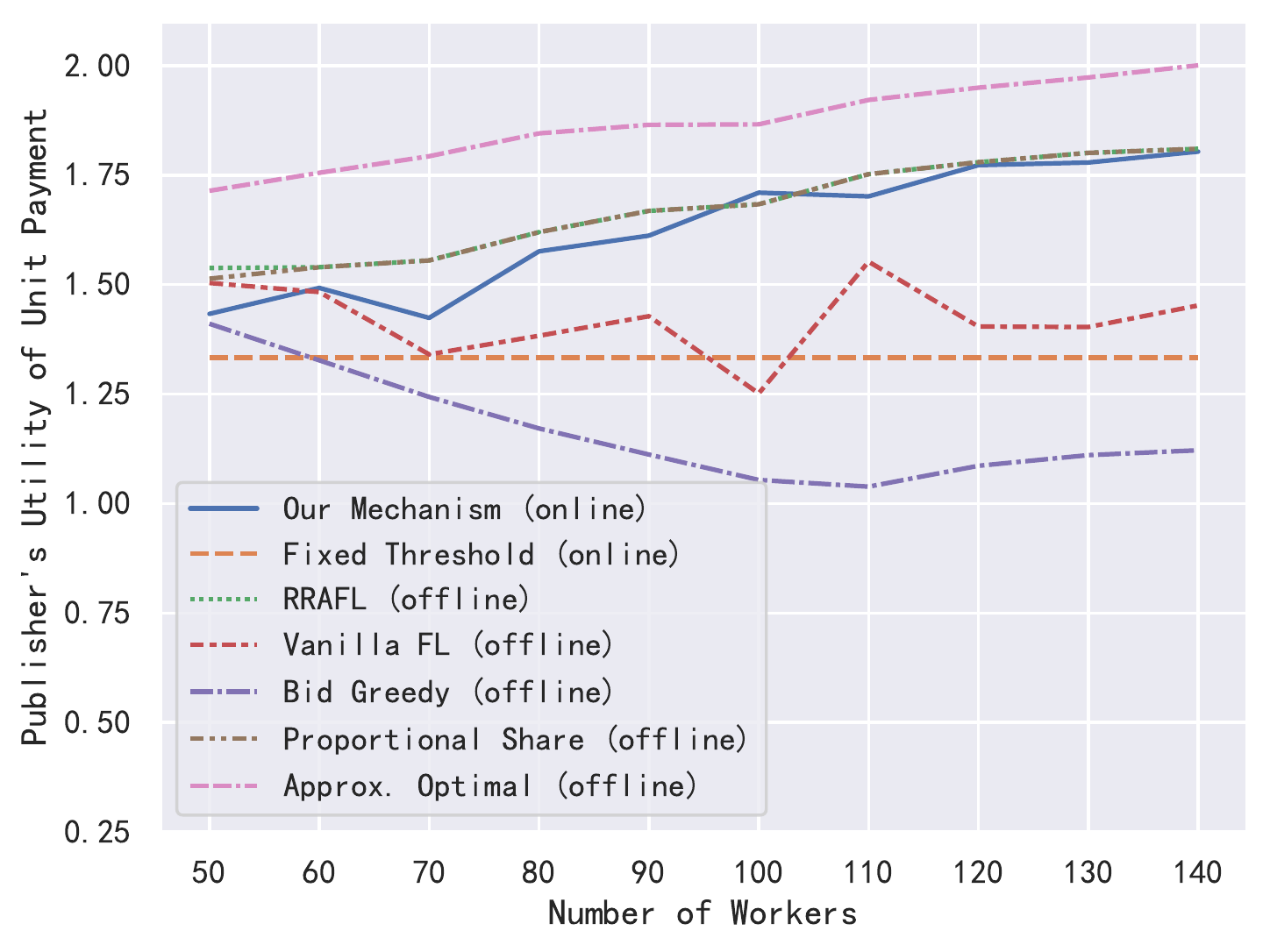}}
  \caption{Impact of the number of workers on the publisher's utility of unit payments}
  \label{fig:worker_cnt_compare}
\end{figure}

First, we set up 100 workers, whose accumulated reputation  $Re_i$, bid price $b_i$ and internal reputation $re_i$ are randomly generated from $[0, 1]$, $[\frac{1}{3}Re_i+\frac{1}{15}, \frac{1}{3}Re_i+\frac{4}{15}]$, $[\max(0, Re_i - 0.1), \min(1, Re_i+0.1)]$, respectively, and the probability of arriving at each time step $t$ is $\frac{1/t}{\sum_ {t}1/t}$. Figure~\ref{fig:budget_compare} shows the effect of budget on the unit payment utility of the task publisher. Figure~\ref{fig:worker_cnt_compare} shows the impact of the number of workers when the budget is 125. The offline Approx. Optimal mechanism outperforms ours because of holding all prior information. Our mechanism achieves cost and time truthfulness by sacrificing utility, but its result is still close to the offline RRAFL and Proportional Share mechanisms, and significantly outperforms other benchmarks. We can infer that our mechanism can improve the unit payment utility of the task publisher as much as possible  in the online scenario.

\begin{table}[tb]
  \centering
  \begin{tabular}{ccc}
  \toprule
  \textbf{Mechanism} & \textbf{MNIST} & \textbf{Fashion MNIST}  \\
  \toprule
  \textbf{Our Mechanism}    & \textbf{1.0000} & \textbf{0.9411}  \\
  \textbf{Fixed Threshold}  & 1.0000 &  0.9175 \\
  \textbf{RRAFL}            & 1.0000 & 0.8956 \\
  \textbf{Vanilla FL}       & 0.4770 & 0.4770 \\
  \textbf{Bid Greedy}       & 0.3296 & 0.3296  \\
  \bottomrule
  \end{tabular}
  
  \caption{Proportion of workers with a data accuracy of 1.0 among the selected workers}
  \label{tab:proportion_acc1}
\end{table}

Then we up set 30 workers. 15, 5, 5, and 5 workers are with data accuracy $dacc$ of 1.0, 0.7, 0.4, and 0.1, respectively. Their bids are generated from $[\frac{1}{3}dacc+\frac{1}{15}, \frac{1}{3}dacc+\frac{4}{15}]$. Run 70 tasks for the MNIST dataset and Fashion MNIST dataset respectively, and use the last 65 tasks to evaluate the mechanism . The budget in each task is 80. From Table~\ref{tab:proportion_acc1}, we can observe that the proportion of workers with a data accuracy of 1.0 selected by our mechanism remains high whether it is the MNIST task or the Fashion MNIST task. We can infer that our mechanism can help select high-quality workers in online scenarios.

\begin{table}[htbp]
  \centering
  \begin{tabular}{ccc}
  \toprule
  \textbf{Mechanism} & \textbf{MNIST} & \textbf{Fashion MNIST}  \\
  \toprule
  \textbf{Our Mechanism}    & 0.4393 & \textbf{1.1162}  \\
  \textbf{Fixed Threshold}  & 0.4993 & 1.4971  \\
  \textbf{RRAFL}            & \textbf{0.4389} & 1.1376 \\
  \textbf{Vanilla FL}       & 0.5022 & 1.4566 \\
  \textbf{Bid Greedy}       & 0.5452 & 1.6925  \\
  \bottomrule
  \end{tabular} 
  
  \caption{Average loss of global models with different mechanisms}
  \label{tab:loss}
\end{table}

Table~\ref{tab:loss} illustrates that the average loss of the global model in our mechanism is close to the offline mechanism with the best results, and substantially lower than that in the offline Vanilla FL and Bid Greedy mechanisms. This suggests that our mechanism can help to improve the model quality in the online scenarios as in the offline scenario.

\section{Conclusion}
We designed an online auction-based incentive mechanism for horizontal federated learning to help the task publisher select and pay workers who arrive one by one online. The task is divided into multiple time steps. At the first time step, workers are selected through the proportional share mechanism. At other time steps, two worker sets are mutually used as sample sets to calculate the payment density thresholds, which are used to make decisions on the two worker sets respectively. Finally, theoretical analysis proves that our online mechanism satisfies six economic properties. The experimental results show its effectiveness.





\bibliographystyle{named}
\bibliography{ijcai22}

\end{document}